\documentclass[10pt, conference, letterpaper]{ IEEEtran}
\IEEEoverridecommandlockouts
\usepackage{cite}
\usepackage{amsmath,amssymb,amsfonts}
\usepackage{amsthm}
\usepackage{booktabs}
\usepackage{graphicx}
\usepackage{textcomp}
\usepackage{xcolor}
\usepackage{xspace}
\usepackage{algorithm}
\usepackage{algpseudocode}
\usepackage{tabularx}
\usepackage{multirow}

\newtheorem{assumption}{Assumption}

\newtheorem{theorem}{Theorem}
\newtheorem{lemma}{Lemma}

\newtheorem{proposition}{Proposition}
\newtheorem{remark}{Remark}

\newcommand{\SPFR}{\textnormal{\textsc{SPFR}}\xspace}
\newcommand{\NoRoute}{\texttt{NO\_SEMANTIC\_ROUTE}\xspace}

\def\BibTeX{{\rm B\kern-.05em{\sc i\kern-.025em b}\kern-.08em
    T\kern-.1667em\lower.7ex\hbox{E}\kern-.125emX}}
\begin{document}

\title{SPFR: Semantic Potential Field Routing for the Distributed Internet of Agents}

\author{
Yeguang Qin,~\IEEEmembership{Student Member, IEEE},
Liangqi Peng,~\IEEEmembership{Student Member, IEEE},
Fengxiao Tang,~\IEEEmembership{Senior Member, IEEE},\\
and Ming Zhao,~\IEEEmembership{Member, IEEE}
}

\maketitle

\begin{abstract}
In a distributed Internet of Agents (IoA) without centralized routing control, routing tasks to capability-matched executors is challenging because destinations are not predetermined and agents have bounded local service views. Discover-then-forward approaches, by contrast, select an executor before network forwarding and therefore do not directly support reselection when additional candidates become visible downstream. 
We introduce Semantic Potential Field Routing (SPFR), a distributed IoA routing algorithm that integrates executor discovery and reselection into hop-by-hop forwarding. SPFR represents each executor visible in a local semantic forwarding information base (FIB) as a task-conditioned semantic potential source, with utility setting its strength and hop distance inducing exponential attenuation. At each hop, the forwarding agent recomputes these potentials, reselects the dominant executor, and forwards the task one hop toward it. Under task-consistent frozen-FIB conditions, we prove loop freedom and
finite-hop termination and derive an explicit additive error bound under
bounded visibility relative to the full-visibility objective. Extensive simulations on real-world topologies show that SPFR approaches
the realized utility of distributed utility-greedy routing and
request-triggered global discovery while using fewer forwarding hops and
substantially fewer request-triggered messages, and remains robust under
network and service dynamics.
\end{abstract}


\section{Introduction}
The Internet of Agents (IoA) is emerging as a framework for interconnecting
heterogeneous autonomous agents for collaborative task
execution~\cite{ioa,iaia}. In this setting, agents span edge and cloud
infrastructures and multiple administrative domains, exposing models, tools,
and services for remote invocation. Supporting such collaboration requires capability-aware service discovery, together with network support for request forwarding and result delivery.

In a distributed IoA without centralized routing control, tasks specify required capabilities rather than fixed executor addresses. Routing therefore involves identifying a suitable executor and delivering the task to it. Keeping every agent informed of all executors' capabilities and changing service states would require frequent network-wide dissemination. When such dissemination is bounded, each agent observes a location-dependent candidate set. The resulting problem is how to coordinate executor selection with network forwarding as candidate visibility changes across locations. 

Existing IoA protocols provide several prerequisites for inter-agent interaction. A2A supports capability advertisement and task exchange~\cite{a2a}, ANP targets cross-domain agent communication~\cite{anp}, and AgentDNS provides agent naming and discovery~\cite{agentdns}. At the application layer, model and agent routers use task-specific information to select among configured models, roles, or agent teams~\cite{routellm,masrouter}. Recent decentralized proposals distribute selection and coordination decisions among individual agents. BiRouter selects collaborators using local information, AMRO-S constructs semantically conditioned agent paths, and AIN outlines capability-oriented routing for open agent networks~\cite{birouter,amros,ain}. Beyond IoA, distributed semantic discovery has associated ontology-indexed service concepts with next-hop, hop-count, and pheromone state in routing tables~\cite{zhangantservice}. Together, these studies span capability representation, candidate resolution, decentralized coordination, and semantic path construction.

At the network layer, QoS and multi-criteria protocols optimize paths to identified destinations, whereas joint communication--computation schemes select services and routes using orchestrator- or controller-wide state~\cite{ddr,ppf,cocar,servicemesh}. Potential-field routing provides a closer technical precedent by using distributed potentials to guide requests toward service instances, gateways, or anycast targets~\cite{fieldservice,heat,pfnsar}. Potential fields have also coordinated task allocation with routing in manufacturing systems~\cite{pachpotential}. These studies demonstrate field-guided service selection and forwarding under their respective service, destination, and state assumptions. However, neither the agent-side nor the network-side line provides a distributed IoA routing algorithm that jointly updates task-dependent executor selection and next-hop forwarding as successive agents observe different local candidates.

To address this gap, we introduce Semantic Potential Field Routing (SPFR), an IoA routing algorithm that integrates executor discovery and reselection with hop-by-hop task forwarding. Bounded-radius dissemination of service descriptors populates a local semantic forwarding information base (FIB) at each agent. SPFR represents each visible eligible executor as a task-conditioned potential source. Its strength reflects semantic fit, projected load, and execution price, while its attraction decays exponentially with hop distance. At each hop, the forwarding agent recomputes these potentials, selects the dominant executor, and forwards the task through the corresponding FIB next hop. Because successive agents consult different local FIBs, candidates outside the source's initial view can enter later forwarding decisions without request-triggered network-wide discovery.  If the source sees no
positive-potential executor, SPFR returns NoRoute; zero-attractor
exploration is outside the present scope. The main contributions of this
paper are as follows:

\begin{itemize}

\item 
We formulate capability-driven IoA task routing as a constrained optimization
problem that jointly selects an executor and request path under execution,
return-route, deadline, and budget constraints. We prove that the resulting
joint optimization problem is NP-hard and derive tractable utility-only and
hop-regularized reference objectives for theoretical analysis.
\item 
To the best of our knowledge, SPFR is the first IoA routing algorithm to
integrate executor discovery and reselection with hop-by-hop forwarding through
semantic potential fields. At each forwarding agent, SPFR computes
task-dependent executor potentials from a bounded semantic FIB and jointly
updates the executor and corresponding next hop using local information. Under
the stated FIB conditions, we prove loop freedom and finite-hop termination.
We further establish global optimality for the hop-regularized objective under
full visibility and derive an explicit additive error bound under bounded
visibility.
\item 
We evaluate SPFR through paired simulations on GEANT, UNINETT, and Deltacom.
SPFR approaches the utility of distributed utility-greedy routing and
request-triggered global discovery with fewer forwarding hops and
substantially less signaling, while remaining robust under network and
service dynamics.
\end{itemize}

\section{Related Work}

\textbf{Agent capability discovery and task routing.}
Agent-interaction protocols support tool access, capability advertisement,
naming, cross-domain communication, and service discovery
~\cite{mcp,a2a,anp,agentdns}. Model and agent routers select among known LLMs
or collectives using task-dependent quality, cost, or difficulty
~\cite{frugalgpt,routellm,automix,irtrouter,iclrouter,masrouter,avengers}.
Directory, semantic-matching, and name-based systems resolve descriptive
requirements before transmission~\cite{ans,usableagentdiscovery,ins,
semanticmatch,dona,ndn}. Distributed schemes further incorporate
ontology-indexed routing state or decentralize collaborator and semantic-path
selection~\cite{zhangantservice,birouter,amros,ain}. These lines do not
explicitly model bounded, location-dependent service views in which a
downstream agent may replace an executor selected upstream.

\textbf{Network-aware and potential-field service routing.}
QoS and multi-criteria protocols route toward identified destinations, while
joint communication--computation schemes commonly rely on controller-wide
service and network state~\cite{wangqos,ddr,ppf,cocar,servicemesh}.
Potential-field routing provides the closest precedent. Lenders \emph{et al.}
model service instances as charges and forward requests toward the neighbor
with the greatest potential~\cite{fieldservice}. Related methods construct
gradients toward gateways or anycast targets and combine service attraction
with network locality~\cite{heat,densityanycast,pfnsar,pachpotential}.
However, they generally assume that relevant targets or gradients are already
available under their dissemination model. \SPFR{} instead jointly updates
executor selection and next-hop forwarding as bounded semantic FIBs reveal
new candidates.

\section{System Model}

\subsection{Network and Agent Model}

We model the IoA communication substrate as a finite, connected, undirected
base graph $G=(V,E)$, where $V=\{v_1,\ldots,v_{|V|}\}$ is the agent set and
each $\{u,v\}\in E$ is a direct underlay adjacency. Every agent may originate
and forward task requests, whereas only a subset advertises execution
services. Service capability is an agent attribute; \emph{executor} is the
task-specific role assigned to a selected and eligible agent.

Each service-capable agent $j$ advertises the descriptor
$\mathbf b_j=(\mathrm{id}_j,\mathcal C_j,\mathrm{state}_j,\tau_j,Q_j,
\mu_j,\pi_j^{\rm unit},\mathrm{seq}_j)$.
Here, $\mathcal C_j=\{\mathbf c_{j,a}\}_{a=1}^{p_j}$ contains its semantic
capability embeddings,
$\mathrm{state}_j\in\{\mathrm{active},\mathrm{inactive}\}$ indicates whether
it accepts tasks, and $Q_j\ge0$ is its advertised queued workload. The
remaining fields specify trust $\tau_j\in[0,1]$, service rate $\mu_j>0$, unit
execution price $\pi_j^{\rm unit}\ge0$, agent identity, and descriptor
sequence number.

The semantic control plane disseminates descriptors within a bounded hop
radius. Cached descriptors may therefore become stale because of propagation
delay, message loss, or subsequent state changes. For each undirected
adjacency $\{u,v\}\in E$, the underlay maintains separate directional
estimates. Direction $u\to v$ is described by
$\mathbf a_{uv}=(b_{uv},\delta_{uv},\kappa_{uv})$, where $b_{uv}>0$ is
available bandwidth, $\delta_{uv}\ge0$ is fixed link latency, and
$\kappa_{uv}\ge0$ is communication cost per data unit. In general,
$\mathbf a_{uv}\ne\mathbf a_{vu}$. For payload size $x\ge0$, the estimated
one-hop delay is $d_{uv}(x)=\delta_{uv}+x/b_{uv}$ and the communication cost
is $c_{uv}(x)=\kappa_{uv}x$. Runtime queueing, jitter, packet loss, and link
degradation may cause realized values to differ from these estimates.

At time $t$, the operational substrate is
$G(t)=(V(t),E(t))$, where $V(t)\subseteq V$ and $E(t)\subseteq E$.
Setting $\mathrm{state}_j=\mathrm{inactive}$ does not remove agent $j$ from
$V(t)$; it only prevents the agent from serving as an executor. Node and link
failures remove the corresponding elements from $V(t)$ and $E(t)$. Because
$G(t)$ need not remain connected, connectivity of $G$ guarantees only nominal
structural reachability, not an operational request path and valid return
route to an eligible executor.

\subsection{Task Semantics and Executor Utility}

\paragraph{Task descriptor}
We represent each task as
$T=(s_T,\mathcal R_T,\ell_T,\boldsymbol{\zeta}_T)$, where $s_T\in V$ is
the source, $\mathcal R_T=\{\mathbf r_b\}_{b=1}^{m_T}$ contains its
$m_T\ge1$ semantic requirements, and $\ell_T\ge0$ is its computational
workload. The task-policy descriptor is
$\boldsymbol{\zeta}_T=(\varsigma_T^{\rm req},\varsigma_T^{\rm ret},
\theta_T,\tau_T^{\min},B_T,D_T^{\max})$.
Here, $\varsigma_T^{\rm req}$ and $\varsigma_T^{\rm ret}$ are the request
and result payload sizes. The threshold $\theta_T$ is the minimum acceptable
similarity for every capability--requirement pair. The remaining fields
specify minimum trust $\tau_T^{\min}$, total budget $B_T$, and maximum
completion delay $D_T^{\max}$.

\paragraph{Semantic matching}
For candidate agent $j$, let
$M_{ab}(j,T)=\operatorname{sim}(\mathbf c_{j,a},\mathbf r_b)\in[0,1]$.
Let $\mathcal X_j(T)$ contain the binary assignment matrices
$\mathbf x=[x_{ab}]$ satisfying $\sum_a x_{ab}=1$ for every requirement
$b$, $\sum_b x_{ab}\le1$ for every capability $a$, and $x_{ab}=0$ whenever
$M_{ab}(j,T)<\theta_T$. If $\mathcal X_j(T)=\emptyset$, we set
$\chi_j^{\rm cap}(T)=S(j,T)=0$. Otherwise,
$\chi_j^{\rm cap}(T)=1$ and
$S(j,T)=m_T^{-1}\max_{\mathbf x\in\mathcal X_j(T)}
\sum_{a,b}x_{ab}M_{ab}(j,T)$.
Thus, eligibility requires distinct capabilities to cover all requirements
above $\theta_T$, while $S(j,T)$ is the maximum average matched similarity.
The assignment is solvable in polynomial time using the Hungarian
method~\cite{kuhn}.

\paragraph{Executor utility}
All agents use the same similarity function, utility weights, and
normalization parameters. They therefore obtain identical executor-side
quantities from the same task and descriptor snapshot. Both $Q_j$ and
$\ell_T$ are measured in computational-work units, while $\mu_j$ is measured
in work units per unit time. Let $\Delta_{\rm ref}>0$ and $\Pi_{\rm ref}>0$
denote the system-wide load and price normalization constants. The estimated
execution delay is $D_j^{\rm exec}(T)=(Q_j+\ell_T)/\mu_j$, and the normalized
projected workload is
$\rho_j^T=\min\{(Q_j+\ell_T)/(\mu_j\Delta_{\rm ref}),1\}$.
The task price is $\pi_j(T)=\pi_j^{\rm unit}\ell_T$, with normalized value
$\widetilde{\pi}_j(T)=\min\{\pi_j(T)/\Pi_{\rm ref},1\}$.

The task-conditioned executor utility is
$U_T(j)=w_sS(j,T)+w_l(1-\rho_j^T)+w_p(1-\widetilde{\pi}_j(T))$, where
$w_s,w_l,w_p\ge0$ and $w_s+w_l+w_p=1$. Hence,
$0\le U_T(j)\le1$. Utility ranks candidates only after the agent-side
eligibility conditions in Sec.~\ref{sec:eligibility} have been evaluated.

\paragraph{Request--return feasibility}
For phase $x\in\{\mathrm{req},\mathrm{ret}\}$, the task-specific link delay
and communication cost are
$d_{uv}^{x}(T)=d_{uv}(\varsigma_T^{x})$ and
$c_{uv}^{x}(T)=c_{uv}(\varsigma_T^{x})$.
Let $P^{\rm req}:s_T\leadsto j$ denote the request path and
$P^{\rm ret}:j\leadsto s_T$ the return path.

The estimated completion delay
$D_T(j;P^{\rm req},P^{\rm ret})$ sums the request-path delay
$\sum_{(u,v)\in P^{\rm req}}d_{uv}^{\rm req}(T)$, execution delay
$D_j^{\rm exec}(T)$, and return delay $D_j^{\rm ret}(s_T)$.
Similarly, the total cost $C_T(j;P^{\rm req},P^{\rm ret})$ sums request
communication cost $\sum_{(u,v)\in P^{\rm req}}c_{uv}^{\rm req}(T)$,
execution price $\pi_j(T)$, and return communication cost
$K_j^{\rm ret}(s_T)$. If either path is unavailable, the executor--path
combination is infeasible. Otherwise, it satisfies the service-level agreement (SLA) only if
$D_T(j;P^{\rm req},P^{\rm ret})\le D_T^{\max}$ and
$C_T(j;P^{\rm req},P^{\rm ret})\le B_T$.
These path-dependent constraints do not alter the executor utility $U_T(j)$.

\subsection{Local Semantic FIB and Reply Interface}
\label{sec:local-semantic-fib}

Let $V_{\rm svc}\subseteq V$ contain the agents that advertise execution
services. The control horizon $H_{\rm ctrl}\in\mathbb Z_{\ge0}$ bounds the
hop radius over which their descriptors are disseminated. The superscript
$H$ denotes the resulting bounded semantic view. This radius limits local
information but does not bound the complete request path.

At the control-plane snapshot considered for a task, let $h(i,j)$ be the
shortest-hop distance from $i$ to $j$ in the active graph, with
$h(i,j)=\infty$ when they are disconnected. Agent $i$ stores each visible
service descriptor and its request-forwarding state as
$FIB_i^H[j]=(\mathbf b_j,h_i(j),next_i(j),timestamp)$.
Here, $h_i(j)$ and $next_i(j)$ are the stored hop distance and request next
hop. For a converged entry, $h_i(j)=h(i,j)$ and $next_i(j)$ lies on a
shortest-hop path to $j$. A service-capable agent maintains a self-entry with
$h_i(i)=0$ and $next_i(i)=i$. The timestamp supports freshness checks, while
$\mathrm{seq}_j$ orders updates generated by agent $j$.

Local visibility and the corresponding visible service set are defined as
\begin{equation*}
\begin{aligned}
vis_i^H(j)
&=\mathbf 1\{j\text{ has a valid entry in }FIB_i^H\},\\[-1mm]
\mathcal V_i^H
&=\{j\in V_{\rm svc}:vis_i^H(j)=1\}.
\end{aligned}
\end{equation*}
Every valid remote entry satisfies $h_i(j)\le H_{\rm ctrl}$. A local FIB is
\emph{TTL-complete} when
\begin{equation*}
\mathcal V_i^H
=
\{j\in V_{\rm svc}:h(i,j)\le H_{\rm ctrl}\}.
\end{equation*}
TTL completeness is an analytical condition, not a prerequisite for making a
local forwarding decision.

The FIB stores no task-specific match, utility, or candidate set; these
quantities are computed after task arrival. It supplies request next hops,
while the underlay provides return-route feasibility and aggregate metrics
through
$\operatorname{ReplyLookup}(j,s_T)\mapsto
(valid_j^{\rm ret}(s_T),K_j^{\rm ret}(s_T),D_j^{\rm ret}(s_T))$.
Here, $valid_j^{\rm ret}(s_T)$ indicates whether an operational return route
exists, while $K_j^{\rm ret}(s_T)$ and $D_j^{\rm ret}(s_T)$ are its estimated
cost and delay. A distributed link-state protocol or an equivalent underlay
service can implement this lookup without performing semantic executor
selection. The lookup reveals no remote service descriptor and therefore
does not enlarge semantic visibility. Structural reachability, return-route
availability, and semantic visibility are consequently distinct.

\subsection{Task-Specific Executor Eligibility}
\label{sec:eligibility}

For analysis, define $\mathcal A_T^0\subseteq V_{\rm svc}$ as the
request-path-independent reference eligibility set. An agent
$j\in V_{\rm svc}$ belongs to $\mathcal A_T^0$ when
$\mathrm{state}_j=\mathrm{active}$,
$\chi_j^{\rm cap}(T)=1$, $\rho_j^T<1$,
$\tau_j\ge\tau_T^{\min}$, and
$valid_j^{\rm ret}(s_T)=1$. It must also satisfy
$\pi_j(T)+K_j^{\rm ret}(s_T)\le B_T$ and
$D_j^{\rm exec}(T)+D_j^{\rm ret}(s_T)\le D_T^{\max}$.

The condition $\chi_j^{\rm cap}(T)=1$ already requires a complete one-to-one
matching in which every selected capability--requirement pair meets
$\theta_T$; no additional threshold on the average score $S(j,T)$ is needed.
The condition $\rho_j^T<1$ excludes agents whose projected workload fills the
service-admission horizon, while load differences below saturation remain
represented in $U_T(j)$. Because the final two conditions omit request-path
delay and cost, they are necessary path-independent screens rather than
sufficient end-to-end feasibility conditions.

The set $\mathcal A_T^0$ is an analytical reference and is not maintained as
global routing state. For the frozen analysis, agent $i$'s visible eligible
set is
$\mathcal E_i^H(T)=\mathcal A_T^0\cap\mathcal V_i^H$. If its FIB is
TTL-complete, then
$\mathcal E_i^H(T)=\{j\in\mathcal A_T^0:
h(i,j)\le H_{\rm ctrl}\}$.

For the frozen analysis, define the gated utility
$\overline U_T(j)=U_T(j)$ for $j\in\mathcal A_T^0$ and
$\overline U_T(j)=0$ otherwise.

At runtime, agent $i$ evaluates
$\widehat{\operatorname{Eligible}}_i(j,T,t)$ using the descriptor stored in
its current FIB, the task parameters, and the return-route lookup. Let
$\widehat S_i(j,T,t)$ and $\widehat U_{T,i}(j,t)$ denote the semantic score
and executor utility evaluated from that stored descriptor. Define the
locally gated utility as
\[
\widehat U_{T,i}^{+}(j,t)
=
\widehat U_{T,i}(j,t)
\widehat{\operatorname{Eligible}}_i(j,T,t).
\]
The implementation never constructs or enumerates $\mathcal A_T^0$. Under
the task-consistent frozen snapshot, for every visible executor,
$\widehat U_{T,i}^{+}(j,t)=\overline U_T(j)$.

The eligibility set captures executor-side and return-side necessary
conditions but does not certify an end-to-end request path. The next section
combines request-path delay and cost with the execution and return-route
quantities.

\section{Problem Formulation}

Problems~$\mathrm{P1}$, $\mathrm{P1}^{\prime}$, and $\mathrm{P2}$ use the
control-plane state observed when task $T$ is admitted. We suppress the
snapshot index. Later control-plane updates affect subsequent tasks but not
the analyzed decision instance. For nonnegative objectives, we adopt
$\max\emptyset=0$.

The model separates executor utility and path-independent eligibility from
request-path feasibility. We first formulate the ideal joint decision, then
introduce a utility-only upper benchmark and a hop-regularized routing
objective.

\subsection{P1: Global Joint Executor--Path Optimization}

For request path $P:s_T\leadsto j$, define its accumulated communication cost
as $C_T^{\rm req}(P)=\sum_{(u,v)\in P}c_{uv}^{\rm req}(T)$ and its delay as
$D_T^{\rm req}(P)=\sum_{(u,v)\in P}d_{uv}^{\rm req}(T)$.

Problem~$\mathrm{P1}$ jointly selects executor $j$ and request path $P$ to
maximize $U_T(j)$~\cite{poularakis,yuanmec}. Feasibility requires $j\in\mathcal A_T^0$ and
$P:s_T\leadsto j$. It further requires
$C_T^{\rm req}(P)+\pi_j(T)+K_j^{\rm ret}(s_T)\le B_T$ and
$D_T^{\rm req}(P)+D_j^{\rm exec}(T)+D_j^{\rm ret}(s_T)
\le D_T^{\max}$. The return route is supplied by the underlay and does not
constitute a second semantic executor-selection decision.

Let \textnormal{\textsc{P1-Feas}} denote whether Problem~$\mathrm{P1}$ has a
feasible executor--path pair. All numerical inputs are nonnegative
finite-precision binary values.

\begin{theorem}
\textnormal{\textsc{P1-Feas}} is NP-complete. Consequently,
Problem~$\mathrm{P1}$ is NP-hard.
\end{theorem}

\begin{proof}
Membership in NP follows because a certificate specifies an executor and a
simple request path. Capability matching, path validity, and all additive
deadline and budget constraints can be verified in polynomial time.

For NP-hardness, reduce from \textsc{Partition}. Given positive integers
$a_1,\ldots,a_n$ with total sum $S$, construct an undirected chain of $n$
diamond subgraphs. Diamond $i$ connects $v_{i-1}$ to $v_i$ through either
$x_i$ or $y_i$. The branch through $x_i$ has total request cost $a_i$ and
delay $2$, whereas the branch through $y_i$ has total request cost $0$ and
delay $2+a_i$. These polynomially representable branch totals follow from the link model by
setting the request payload to one and the result payload, execution delay,
and price to zero. Set $s_T=v_0$ and make $v_n$ the only eligible executor. Let
\[
B_T=S/2,\qquad D_T^{\max}=2n+S/2.
\]
Every simple $v_0$--$v_n$ path selects exactly one branch from each diamond.
If the integers assigned to the $x_i$ branches sum to $X$, its request cost
is $X$ and its request delay is $2n+S-X$. Hence, both constraints are
satisfied if and only if
\[
X\le S/2
\quad\text{and}\quad
S-X\le S/2,
\]
which holds exactly when $X=S/2$. Therefore, the constructed
Problem~$\mathrm{P1}$ instance is feasible if and only if the
\textsc{Partition} instance is feasible. The construction is polynomial.
\end{proof}

\subsection{P1$^{\prime}$: Utility-Only Upper Benchmark}

Problem~$\mathrm{P1}^{\prime}$ isolates executor ranking from constrained-path
search while retaining the path-independent eligibility defined in
Sec.~\ref{sec:eligibility}. Its optimum is
$U_{\rm sel}^*(T)=\max_{j\in\mathcal A_T^0}U_T(j)$.
For each feasible Problem~$\mathrm{P1}$ instance, let
$U_{\mathrm{P1}}^*(T)$ denote the optimal utility. Every
Problem~$\mathrm{P1}$-feasible executor belongs to $\mathcal A_T^0$, so
$U_{\rm sel}^*(T)\ge U_{\mathrm{P1}}^*(T)$.

Thus, $U_{\rm sel}^*(T)$ is a polynomial-time global upper benchmark on the
snapshot utility attainable by Problem~$\mathrm{P1}$. It requires only
polynomial-time capability matching and candidate scanning, without
constrained-path search.

\subsection{P2: Hop-Regularized Semantic Routing Objective}

Problem~$\mathrm{P1}^{\prime}$ has no preference for executors near the
current routing location. We therefore define the global snapshot objective
$\Psi_s^*(T)=\max_{j\in\mathcal A_T^0}
\overline U_T(j)\exp[-\omega_h h(s_T,j)]$, where $\omega_h>0$ controls hop
attenuation and $h$ is the shortest-hop distance. The subscript $s$ denotes
the source $s_T$.

Problem~$\mathrm{P2}$ is a global reference objective. Once eligibility and
utility are available, one breadth-first search computes all source distances
in $O(|V|+|E|)$, followed by a linear candidate scan.

At forwarding agent $i$, \SPFR{} evaluates the local counterpart of
Problem~$\mathrm{P2}$ over $\mathcal E_i^H(T)$ using the stored distance
$h_i(j)$. Both the candidate set and distance origin may therefore change
between hops. Section~\ref{sec:visibility-analysis} proves exact realization
under full visibility and derives an additive bounded-view error.

Problem~$\mathrm{P2}$ retains only the path-independent conditions in
$\mathcal A_T^0$. It does not evaluate accumulated request-path cost or delay
and therefore does not certify Problem~$\mathrm{P1}$ feasibility.
Accordingly, \SPFR{} is not claimed as a polynomial-time approximation
algorithm for Problem~$\mathrm{P1}$.

If some eligible executor has positive utility, Problem~$\mathrm{P2}$ has
the same maximizers as
$\max_{j\in\mathcal A_T^0:U_T(j)>0}
[\log U_T(j)-\omega_h h(s_T,j)]$.
Hence, $\omega_h$ controls the trade-off between log-utility and hop locality.
The hop term is a locality regularizer, not an estimate of physical latency
or communication cost.

\begin{proposition}[Segment-composable attenuation]
\label{prop:exponential-kernel}
Let $g:\mathbb Z_{\ge0}\rightarrow(0,1]$ satisfy $g(0)=1$,
$g(1)=\gamma\in(0,1)$, and
$g(h_1+h_2)=g(h_1)g(h_2)$. Then
$g(h)=\gamma^h=\exp(-\omega_h h)$, where
$\omega_h=-\log\gamma>0$.
\end{proposition}

\begin{proof}
The result holds at $h=0$. Multiplicative composition gives
$g(h+1)=g(h)g(1)$. Starting from $g(0)=1$, induction yields
$g(h)=\gamma^h$. Substituting $\omega_h=-\log\gamma$ gives
$g(h)=\exp(-\omega_h h)$.
\end{proof}

Proposition~\ref{prop:exponential-kernel} motivates exponential attenuation
through segment composability: attenuation over concatenated path segments
equals the product of their individual attenuation factors. We do not claim
that this form empirically dominates every alternative decay function.
Instead, it provides a consistent potential transformation for additive hop
distances, while $\omega_h$ controls the empirical utility--locality
trade-off.

For $\overline U_T(j)>0$, define the effective ranking distance as
$d_{\rm eff}(i,j,T)=h(i,j)-\omega_h^{-1}\log\overline U_T(j)$; set
$d_{\rm eff}(i,j,T)=+\infty$ when $\overline U_T(j)=0$. For a fixed task and
forwarding location, maximizing semantic potential is equivalent to minimizing
$d_{\rm eff}$. This quantity is a task-dependent ranking transformation, not
a physical underlay distance or graph metric.

\section{SPFR Design}

Solving Problem~$\mathrm{P2}$ centrally requires global knowledge of all
eligible executors and their source-to-executor hop distances. \SPFR instead
evaluates a node-local counterpart using valid entries in each agent's bounded
semantic FIB. Figure~\ref{fig:spfr-discovery} illustrates how a downstream
forwarding agent discovers a stronger executor and updates both the selected
executor and the next hop.

\begin{figure}[t]
    \centering
    \includegraphics[width=\columnwidth]
    {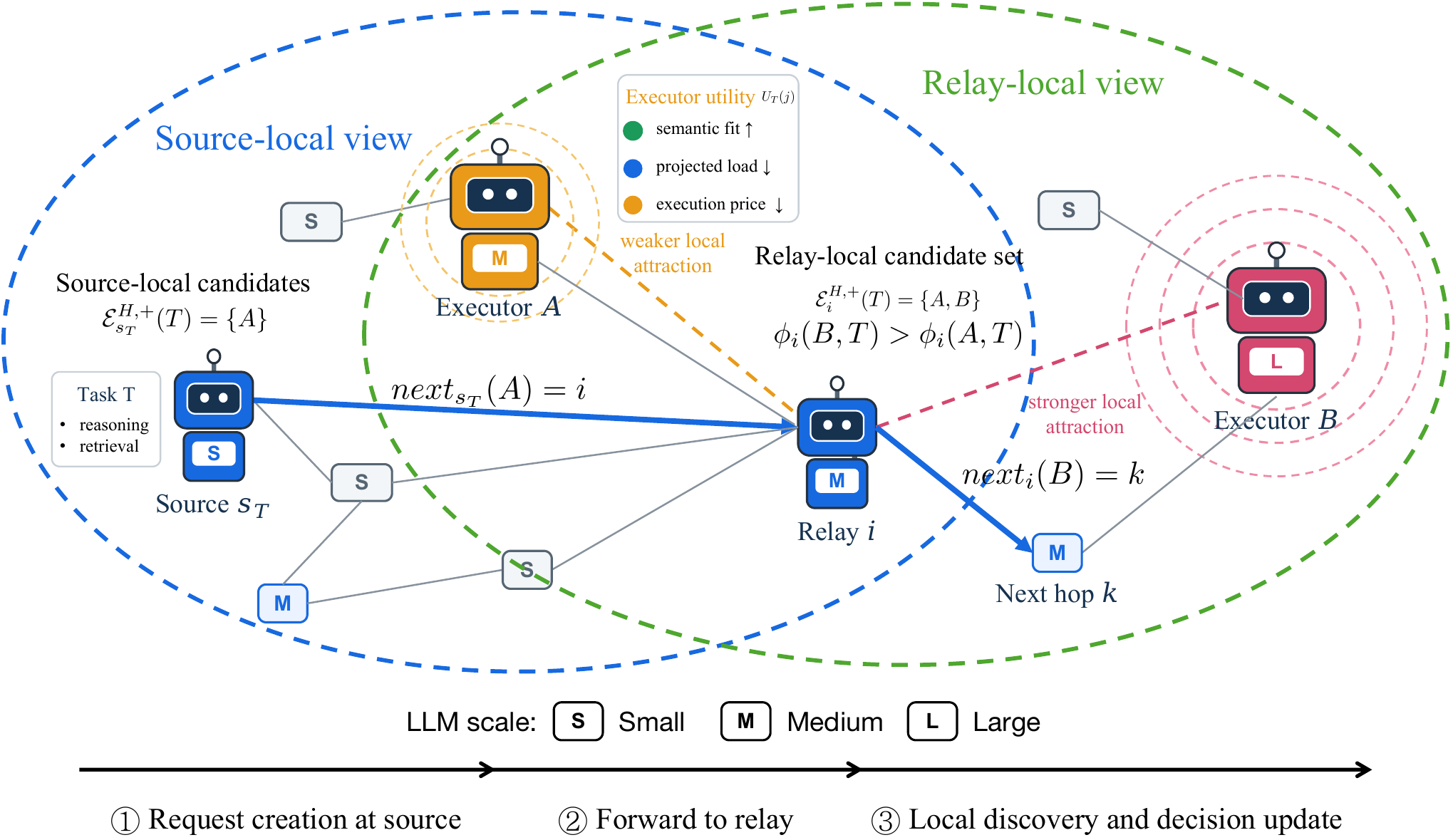}
    \caption{Source-local executor selection and in-path reselection in SPFR.}
    \label{fig:spfr-discovery}
\end{figure}

\subsection{Local Semantic Potential}
At runtime, each visible executor $j$ induces the locally estimated potential
$\widehat\phi_i(j,T,t)=\widehat U_{T,i}^{+}(j,t)
\exp[-\omega_h h_i(j)]$. Agent $i$ considers
$\widehat{\mathcal E}_i^{H,+}(T,t)=
\{j\in\mathcal V_i^H:\widehat U_{T,i}^{+}(j,t)>0\}$ and computes
$\widehat\Phi_i^H(T,t)=
\max_{j\in\widehat{\mathcal E}_i^{H,+}(T,t)}
\widehat\phi_i(j,T,t)$, with $\max\emptyset=0$.

Under the task-consistent frozen snapshot, the local estimates coincide with
their analytical counterparts. We therefore use
$\phi_i(j,T)=\overline U_T(j)\exp[-\omega_h h_i(j)]$,
$\mathcal E_i^{H,+}(T)=
\{j\in\mathcal E_i^H(T):\overline U_T(j)>0\}$, and
$\Phi_i^H(T)=
\max_{j\in\mathcal E_i^{H,+}(T)}\phi_i(j,T)$ in the theoretical analysis.
Thus, the hatted quantities define the locally executable runtime rule,
whereas their un-hatted counterparts are used in the frozen-snapshot proofs.

\subsection{Hop-by-Hop Forwarding and Complexity}

If $\widehat{\mathcal E}_i^{H,+}(T,t)\ne\emptyset$, agent $i$ selects
$j_i^{\rm dom}\in\operatorname*{arg\,max}_{j\in
\widehat{\mathcal E}_i^{H,+}(T,t)}\widehat\phi_i(j,T,t)$. Ties are broken deterministically by executor identifier. If
$j_i^{\rm dom}=i$, agent $i$ executes the task; otherwise, it forwards the
task to $next_i(j_i^{\rm dom})$, where the next agent repeats the decision
using its own semantic FIB. If no positive-potential candidate exists, \SPFR
returns \NoRoute without performing semantic frontier exploration.

The forwarding rule operates on the valid local FIB available at each
decision. Section~\ref{sec:analytical-conditions} states the consistency and
completeness conditions required by the formal guarantees.

Algorithm~\ref{alg:spfr} summarizes the SPFR data-plane procedure. Each
iteration is executed by the current forwarding agent using only its local
semantic FIB and estimated runtime state.

\begin{algorithm}[t]
\caption{SPFR forwarding algorithm}
\label{alg:spfr}
\begin{algorithmic}[1]
\Require Task $T$ at source $s_T$; hop budget $H_{\max}$
\State $i\gets s_T$; $q\gets0$
\While{$q\le H_{\max}$}
    \State $t\gets\operatorname{Now}()$
    \State $\widehat{\mathcal E}_i^{H,+}(T,t)\gets
    \{j\in\mathcal V_i^H:
    \widehat U_{T,i}^{+}(j,t)>0\}$
    \If{$\widehat{\mathcal E}_i^{H,+}(T,t)=\emptyset$}
        \State \Return \NoRoute
        \Comment{no locally eligible attractor}
    \EndIf
    \State $j_i^{\rm dom}\gets
    \operatorname*{arg\,max}_{j\in
    \widehat{\mathcal E}_i^{H,+}(T,t)}
    \widehat\phi_i(j,T,t)$
    \State \textbf{break ties by} executor identifier
    \If{$j_i^{\rm dom}=i$}
        \State \Return $\operatorname{ExecuteAndReply}(T,i,s_T)$
    \EndIf
    \If{$q\ge H_{\max}$}
        \State \Return \NoRoute
        \Comment{hop budget exhausted}
    \EndIf
    \State $k\gets next_i(j_i^{\rm dom})$
    \If{$k=\bot$ \textbf{or} $\{i,k\}\notin E(t)$}
        \State \Return \NoRoute
        \Comment{stale or unavailable next hop}
    \EndIf
    \If{$\operatorname{Forward}(T,i,k)=\mathrm{failure}$}
        \State \Return \NoRoute
        \Comment{runtime forwarding failure}
    \EndIf
    \State $i\gets k$; $q\gets q+1$
\EndWhile
\State \Return \NoRoute
\end{algorithmic}
\end{algorithm}

Under a task-consistent frozen snapshot,
$\widehat U_{T,i}^{+}(j,t)=\overline U_T(j)$ and
$\widehat\phi_i(j,T,t)=\phi_i(j,T)$ for every visible executor. Hence,
Algorithm~\ref{alg:spfr} reduces to the frozen rule analyzed in the following
theorems.

Here, $q$ counts nonterminal request-forwarding hops. Because the executor
check precedes the hop-budget check, a task reaching its executor after
$H_{\max}$ forwarding hops may still execute. We set
$H_{\max}=|V|-1$, matching the finite-termination bound under the
task-consistent frozen snapshot. If the stored next hop is missing or its
link is inactive in $E(t)$, \SPFR returns \NoRoute. The
$\operatorname{Forward}$ operation includes the underlay retry policy;
exhausting its retry limit returns $\mathrm{ForwardFailure}$ and does not
trigger semantic reselection. Under asynchronous control-plane updates,
$H_{\max}$ serves only as a runtime safeguard and does not extend the
frozen-state loop-freedom guarantee.

The per-hop computational cost is dominated by capability matching. For a
visible agent $j$, rectangular Hungarian matching requires
$O(p_jm_T\min\{p_j,m_T\})$ time, while padding the assignment matrix yields
the familiar $O(\max\{p_j,m_T\}^3)$ bound. Assuming constant-time similarity
access and a locally available ReplyLookup result, one decision at agent $i$
requires
$O(\sum_{j\in\mathcal V_i^H}p_jm_T\min\{p_j,m_T\}
+|\mathcal V_i^H|)$ time. The sum covers all visible descriptors because
eligibility is determined only after semantic matching and admission checks.
Bounded dissemination confines the state and computation to the
$H_{\rm ctrl}$-hop neighborhood, although its size may still increase with
local network density.

\section{Theoretical Analysis}

We show that candidate inheritance induces strict dominant-potential ascent,
which establishes loop freedom and finite termination. We then bound the
terminal executor's source-side $\mathrm{P2}$ gap relative to full visibility.

\subsection{Analytical Conditions}
\label{sec:analytical-conditions}

The following conditions support the analysis but are not required merely to
evaluate the local forwarding rule. The attenuation parameter $\omega_h>0$
remains fixed throughout.

\begin{assumption}[Task-consistent snapshot]
\label{ass:task-snapshot}
For each analyzed task $T$, the graph, agent descriptors, utility inputs,
semantic-matching results, reply descriptors, hop counts, and next hops are
obtained from one consistent control-plane snapshot and remain fixed during
the analysis.
\end{assumption}

\begin{assumption}[TTL completeness and next-hop closure]
\label{ass:ttl-closure}
For every agent $i$, the semantic FIB contains every advertised service
descriptor within $H_{\rm ctrl}$ hops, i.e.,
$\mathcal V_i^H=\{j\in V:h(i,j)\le H_{\rm ctrl}\}$, with the correct
shortest-hop distance. Moreover, if $j\ne i$ and $k=next_i(j)$, then
$vis_k^H(j)=1$ and $h_k(j)=h_i(j)-1$.
\end{assumption}

\begin{assumption}[Positive source attractor]
\label{ass:positive-source}
The source observes at least one eligible executor with positive utility,
i.e., $\mathcal E_{s_T}^{H,+}(T)\ne\emptyset$.
\end{assumption}

\subsection{Candidate Inheritance and Potential Ascent}

\begin{lemma}[Candidate inheritance]
\label{lem:inheritance}
Under Assumptions~\ref{ass:task-snapshot} and~\ref{ass:ttl-closure}, let
$j\in\mathcal E_i^H(T)$, $j\ne i$, and $k=next_i(j)$. Then
$j\in\mathcal E_k^H(T)$ and $h_k(j)=h_i(j)-1$. Moreover,
$j\in\mathcal E_i^{H,+}(T)$ implies $j\in\mathcal E_k^{H,+}(T)$.
\end{lemma}

\begin{proof}
Next-hop closure provides a valid entry for $j$ at $k$ and gives
$h_k(j)=h_i(j)-1\le H_{\rm ctrl}-1$. Executor eligibility, task utility, and
the reply descriptor are node-invariant within the task-consistent snapshot.
All eligibility conditions and the positive gated utility are therefore
preserved at $k$.
\end{proof}

\begin{theorem}[Strict dominant-potential ascent]
\label{thm:strict-ascent}
Under Assumptions~\ref{ass:task-snapshot} and~\ref{ass:ttl-closure}, every
nonterminal \SPFR step $i\to k$ satisfies
$\Phi_k^H(T)\ge\exp(\omega_h)\Phi_i^H(T)>\Phi_i^H(T)$.
\end{theorem}

\begin{proof}
Let $j=j_i^{\rm dom}$. The forwarding rule gives
$\phi_i(j,T)=\Phi_i^H(T)>0$. By Lemma~\ref{lem:inheritance}, $j$ remains a
positive candidate at $k$, while its remaining hop count decreases by one.
Consequently,
$\Phi_k^H(T)\ge\overline U_T(j)\exp[-\omega_h(h_i(j)-1)]
=\exp(\omega_h)\phi_i(j,T)=\exp(\omega_h)\Phi_i^H(T)$.
The inequality is strict because $\omega_h>0$ and $\Phi_i^H(T)>0$.
\end{proof}

\subsection{Loop Freedom and Finite Termination}

\begin{theorem}[Loop freedom and finite termination]
\label{thm:loop-free}
Under Assumptions~\ref{ass:task-snapshot}--\ref{ass:positive-source}, \SPFR
is loop-free and terminates at an eligible executor after at most $|V|-1$
nonterminal forwarding steps.
\end{theorem}

\begin{proof}
If an agent were revisited, the task-consistent snapshot would give it the
same FIB, positive candidate set, utility values, and dominant local
potential as before, contradicting Theorem~\ref{thm:strict-ascent}. Hence,
the forwarding trajectory cannot revisit an agent. Candidate inheritance
also preserves at least one positive candidate after every nonterminal step.
The trajectory is therefore a simple path in a finite graph, contains at most
$|V|-1$ edges, and can terminate only when the current agent selects itself
as the dominant eligible executor.
\end{proof}

\subsection{Full and Bounded Visibility}
\label{sec:visibility-analysis}

\begin{theorem}[Full-visibility $\mathrm{P2}$ optimality]
\label{thm:full-p2}
Under Assumptions~\ref{ass:task-snapshot}--\ref{ass:positive-source}, suppose
that $\Psi_s^*(T)>0$ and $H_{\rm ctrl}\ge\operatorname{diam}(G)$. If \SPFR
terminates at executor $m$, then
$m\in\operatorname*{arg\,max}_{j\in\mathcal A_T^0}
\overline U_T(j)\exp[-\omega_h h(s_T,j)]$.
\end{theorem}

\begin{proof}
Full visibility gives $\mathcal V_{s_T}^H=V$,
$\mathcal E_{s_T}^H(T)=\mathcal A_T^0$, and
$\Phi_{s_T}^H(T)=\Psi_s^*(T)$. Let $L_T$ be the number of nonterminal
forwarding steps. At termination, $m$ selects itself, so
$\Phi_m^H(T)=\phi_m(m,T)=\overline U_T(m)$. Repeated application of
Theorem~\ref{thm:strict-ascent} gives
$\overline U_T(m)\ge\exp(\omega_h L_T)\Psi_s^*(T)$.

Because $h(s_T,m)\le L_T$, the terminal source-side potential satisfies
$\overline U_T(m)\exp[-\omega_h h(s_T,m)]
\ge\overline U_T(m)\exp(-\omega_h L_T)\ge\Psi_s^*(T)$.
This value is feasible for Problem~$\mathrm{P2}$ and therefore cannot exceed
$\Psi_s^*(T)$. Equality follows, proving the claim.
\end{proof}

\begin{theorem}[Bounded-visibility additive $\mathrm{P2}$ error]
\label{thm:bounded-p2}
Assume the task-consistent snapshot in Assumption~\ref{ass:task-snapshot},
$0\le\overline U_T(j)\le U_{\max}$, and $\Psi_s^*(T)>0$. Suppose that the
source FIB is TTL-complete within $H_{\rm ctrl}$ hops, every forwarding step
satisfies next-hop closure, and
$\mathcal E_{s_T}^{H,+}(T)\ne\emptyset$.

Define the best source-visible potential as
$V_T^*=\max\{\overline U_T(j)\exp[-\omega_h h(s_T,j)]:
j\in\mathcal A_T^0,\ h(s_T,j)\le H_{\rm ctrl}\}$. If \SPFR terminates at
executor $m$, then
$\overline U_T(m)\exp[-\omega_h h(s_T,m)]\ge V_T^*$.
Furthermore,
$0\le\mathrm{P2Gap}_T\le
U_{\max}\exp[-\omega_h(H_{\rm ctrl}+1)]$, where
$\mathrm{P2Gap}_T=\Psi_s^*(T)-
\overline U_T(m)\exp[-\omega_h h(s_T,m)]$.
\end{theorem}

\begin{proof}
Let
$a_s\in\operatorname*{arg\,max}_{j\in\mathcal E_{s_T}^{H,+}(T)}
\phi_{s_T}(j,T)$, and let $L_T$ be the realized trajectory length. Strict
ascent and terminal self-selection give
$\overline U_T(m)\ge
\exp(\omega_h L_T)\phi_{s_T}(a_s,T)$. Because $h(s_T,m)\le L_T$, it follows
that
$\overline U_T(m)\exp[-\omega_h h(s_T,m)]
\ge\phi_{s_T}(a_s,T)=V_T^*$, where the equality follows from source-side TTL
completeness.

Define the best source-invisible potential as
$I_T^*=\max\{\overline U_T(j)\exp[-\omega_h h(s_T,j)]:
j\in\mathcal A_T^0,\ h(s_T,j)>H_{\rm ctrl}\}$, with the maximum of an empty
set defined as zero. Since hop distance is integer-valued, every
source-invisible executor satisfies $h(s_T,j)\ge H_{\rm ctrl}+1$, and hence
$I_T^*\le U_{\max}\exp[-\omega_h(H_{\rm ctrl}+1)]$.

By construction, $\Psi_s^*(T)=\max\{V_T^*,I_T^*\}$. The terminal value is at
least $V_T^*$ and, because $m\in\mathcal A_T^0$, cannot exceed
$\Psi_s^*(T)$. Therefore,
$0\le\mathrm{P2Gap}_T\le
\max\{V_T^*,I_T^*\}-V_T^*\le I_T^*
\le U_{\max}\exp[-\omega_h(H_{\rm ctrl}+1)]$.
\end{proof}

For tasks with $\Psi_s^*(T)>0$, define
$\mathrm{P2Ratio}_T=
\overline U_T(m)\exp[-\omega_h h(s_T,m)]/\Psi_s^*(T)$ when \SPFR terminates
at $m$, and $\mathrm{P2Ratio}_T=0$ when it returns \NoRoute. Assigning zero
to \NoRoute is an all-task evaluation convention rather than a theoretical
guarantee. When $\Psi_s^*(T)=0$, the ratio is undefined and excluded from
ratio statistics.

For tasks covered by Theorem~\ref{thm:bounded-p2}, define
$\mathrm{NormP2Gap}_T=\mathrm{P2Gap}_T/
\{U_{\max}\exp[-\omega_h(H_{\rm ctrl}+1)]\}$. The theorem gives
$0\le\mathrm{NormP2Gap}_T\le1$, whereas
Theorem~\ref{thm:full-p2} gives $\mathrm{P2Ratio}_T=1$ under full visibility.

\begin{remark}[Scope of the guarantees]
Theorem~\ref{thm:full-p2} establishes exact realization of
Problem~$\mathrm{P2}$ only under full semantic visibility.
Theorem~\ref{thm:bounded-p2} bounds the source-side $\mathrm{P2}$ potential,
not executor utility or the approximation ratio for Problem~$\mathrm{P1}$.
If the source observes no eligible executor with positive utility, \SPFR
returns \NoRoute. The forwarding rule remains operational when the FIB
changes during forwarding, but the frozen-snapshot guarantees no longer
apply.
\end{remark}

\section{Evaluation}
\subsection{Setup}
We implement a Python simulator for distributed task forwarding under dynamic
executor and network states. All methods are evaluated on GEANT, UNINETT, and
Deltacom, with one agent hosted at each node. The service catalog contains 24
atomic capabilities grouped into eight semantic domains. Small, medium, and
large agents constitute $55\%$, $30\%$, and $15\%$ of the population and
advertise 5, 10, and 15 capabilities, respectively. The tiers differ in
semantic coverage, service rate, load, and price. Their task-independent
core--edge placement includes both regional and cross-region capabilities.

Light, standard, and complex tasks occur with probabilities $0.45$, $0.40$,
and $0.15$ and require 2, 2--3, and 3--4 capabilities, respectively.
Computational workload, request size, deadline, and budget span $[0.4,3.5]$,
$[0.15,1.40]$, $[2.5,14.0]$, and $[1.5,24.0]$, respectively. The Hungarian
algorithm performs one-to-one requirement--capability matching, excluding
pairs below the $0.55$ relevance threshold. Task sources and requirements are
sampled independently of executor visibility and routing outcomes.

Unless stated otherwise, $H_{\rm ctrl}=2$, $\omega_h=0.08$, $B=3.5$, the
Poisson arrival rate is $0.8$, and
$(w_s,w_l,w_p)=(0.70,0.15,0.15)$. We set $H_{\max}=|V|-1$, consistent with
the frozen-state loop-freedom bound. The main experiment includes dynamic
queues, service departures and recoveries, link perturbations, background
traffic, and stale control-plane state. Each topology uses 10 paired seeds and
200 tasks per seed. Within each topology--seed block, all methods replay the
same catalog, workload, and exogenous events. Cross-topology results first
average the three topologies within each seed. We then compute paired method
contrasts and two-sided 95\% Student-$t$ confidence intervals across seeds,
without treating tasks within a seed as independent replicates.

We compare RAND, D-SEM, D-GREEDY, \SPFR{}, and GLOBAL$^\dagger$. All methods
apply the same estimated-eligibility predicate within their available views.
RAND, D-SEM, D-GREEDY, and \SPFR{} share the same bounded local FIB, whereas
GLOBAL$^\dagger$ performs request-triggered full discovery. RAND fixes a
randomly selected eligible source-FIB executor. D-SEM and D-GREEDY reselect at
each hop by maximizing $S(j,T)$ and $\overline U_T(j)$, respectively, whereas
\SPFR{} maximizes semantic potential. GLOBAL$^\dagger$ fixes the full-view
potential maximizer and serves only as a high-information reference. A method
returns \NoRoute{} when its post-gate candidate set is empty. Tasks are not
prefiltered by source visibility or routing outcome.

Let $Y_T=1$ when task $T$ completes execution, returns its result, and
satisfies the realized semantic, deadline, and budget requirements; otherwise,
$Y_T=0$. We report all-task realized \emph{Utility} as
$\overline U^{\rm real}=|\mathcal T|^{-1}
\sum_{T\in\mathcal T}Y_TU_T(e_T)$, so failed and \NoRoute{} tasks contribute
zero. Success is $|\mathcal T|^{-1}\sum_{T\in\mathcal T}Y_T$.
Request-triggered messages include task forwarding, result return, and
on-demand discovery and are normalized over all tasks. Semantic similarity,
load, price, hops, P95 delay, and communication cost are conditioned on
successful tasks. Periodic semantic beacons are audited separately as shared
control-plane overhead. Unless stated otherwise, the remaining experiments
use UNINETT.

\begin{table}[t]
\centering
\caption{Cross-topology performance over 10 paired seeds. Utility, Success,
and Msgs./task cover all tasks; other metrics cover successful tasks.}
\label{tab:overall-performance}
\scriptsize
\setlength{\tabcolsep}{0pt}
\renewcommand{\arraystretch}{0.92}
\begin{tabularx}{0.97\columnwidth}{
@{}l
*{8}{>{\centering\arraybackslash}X}@{}}
\toprule
\multirow{2}{*}{Method}
& \multirow{2}{*}{Utility}
& Success
& Msgs.
& Sem.
& \multirow{2}{*}{Load}
& Price
& \multirow{2}{*}{Hops}
& P95 \\
& & (\%) & /task & sim. & & (c.u.) & & (s) \\
\midrule
\multicolumn{9}{c}{\textit{GEANT (40 nodes, 61 links)}} \\
\cmidrule(lr){1-9}
RAND    & 0.679 & 91.8 & 4.08 & 0.769 & 0.361 & 2.144 & 1.70 & 2.69 \\
D-SEM   & 0.711 & 88.0 & 5.24 & 0.920 & 0.452 & 3.264 & 2.41 & 3.97 \\
D-GREEDY & 0.764 & 90.3 & 5.35 & 0.884 & 0.288 & 1.391 & 2.43 & 2.11 \\
SPFR    & 0.751 & 91.1 & 3.87 & 0.867 & 0.324 & 1.599 & 1.59 & 2.03 \\
GLOBAL$^\dagger$  & 0.766 & 92.3 & 136.92 & 0.875 & 0.324 & 1.619 & 1.67 & 2.07 \\
\midrule
\multicolumn{9}{c}{\textit{UNINETT (74 nodes, 101 links)}} \\
\cmidrule(lr){1-9}
RAND    & 0.675 & 90.9 & 3.96 & 0.773 & 0.365 & 2.125 & 1.65 & 2.88 \\
D-SEM   & 0.710 & 87.9 & 5.42 & 0.920 & 0.458 & 3.249 & 2.53 & 4.16 \\
D-GREEDY & 0.756 & 89.7 & 5.15 & 0.883 & 0.304 & 1.436 & 2.34 & 2.13 \\
SPFR    & 0.746 & 90.9 & 3.78 & 0.863 & 0.328 & 1.655 & 1.55 & 2.07 \\
GLOBAL$^\dagger$  & 0.761 & 92.0 & 277.34 & 0.874 & 0.328 & 1.650 & 1.67 & 2.07 \\
\midrule
\multicolumn{9}{c}{\textit{Deltacom (113 nodes, 161 links)}} \\
\cmidrule(lr){1-9}
RAND    & 0.671 & 90.1 & 3.90 & 0.762 & 0.342 & 1.861 & 1.63 & 2.43 \\
D-SEM   & 0.712 & 88.1 & 5.35 & 0.899 & 0.417 & 2.895 & 2.48 & 4.23 \\
D-GREEDY & 0.752 & 89.8 & 4.99 & 0.873 & 0.296 & 1.423 & 2.24 & 2.17 \\
SPFR    & 0.743 & 91.0 & 3.63 & 0.849 & 0.310 & 1.500 & 1.46 & 1.93 \\
GLOBAL$^\dagger$  & 0.767 & 93.2 & 501.01 & 0.857 & 0.310 & 1.497 & 1.59 & 2.00 \\
\bottomrule
\end{tabularx}
\vspace{-1mm}
\end{table}

\begin{figure}[t]
\centering
\begin{minipage}[t]{0.49\columnwidth}
\vspace{0pt}
\centering
\includegraphics[width=\linewidth]{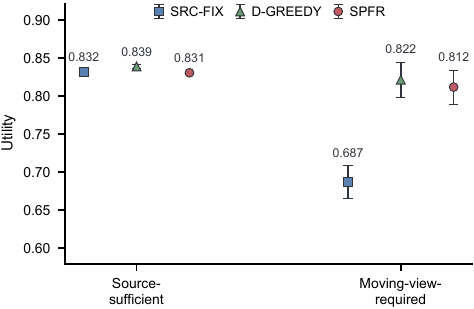}
\par\vspace{0.4mm}
{\footnotesize (a) Utility\par}
\end{minipage}\hfill
\begin{minipage}[t]{0.49\columnwidth}
\vspace{0pt}
\centering
\includegraphics[width=\linewidth]{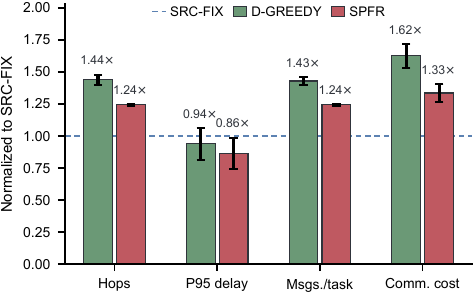}
\par\vspace{0.4mm}
{\footnotesize (b) Network expenditure\par}
\end{minipage}
\vspace{-1mm}
\caption{Benefit and network expenditure of in-path reselection on UNINETT.
(a) Utility for source-sufficient and moving-view-required tasks.
(b) Completed-task network expenditure normalized to SRC-FIX.}
\label{fig:moving-view-tradeoff}
\end{figure}

\begin{figure*}[htbp]
\centering
\begin{minipage}[t]{0.49\textwidth}
\vspace{0pt}
\centering
\includegraphics[width=\linewidth]{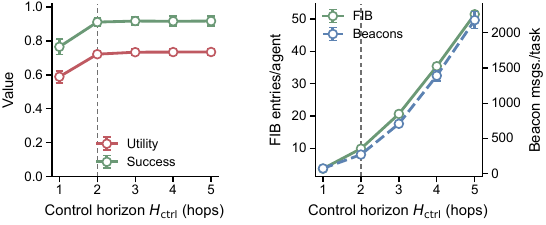}
\par\vspace{0.5mm}
{\footnotesize (a) Control-horizon sensitivity\par}
\end{minipage}\hfill
\begin{minipage}[t]{0.49\textwidth}
\vspace{0pt}
\centering
\includegraphics[width=\linewidth]{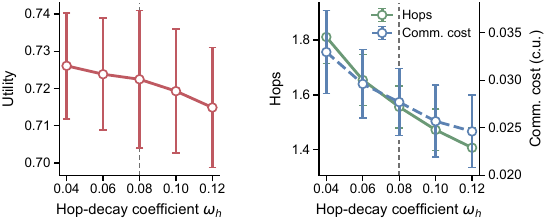}
\par\vspace{0.5mm}
{\footnotesize (b) Hop-decay sensitivity\par}
\end{minipage}
\vspace{-1mm}
\caption{Sensitivity to (a) the control horizon $H_{\rm ctrl}$ and
(b) the hop-decay coefficient $\omega_h$.}
\label{fig:parameter-sensitivity}
\vspace{-2mm}
\end{figure*}

\begin{figure}[t]
\centering
\includegraphics[width=0.4\columnwidth]{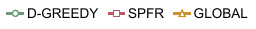}
\par\vspace{-3mm}
\begin{minipage}[t]{0.49\columnwidth}
\vspace{0pt}
\centering
\includegraphics[width=\linewidth]{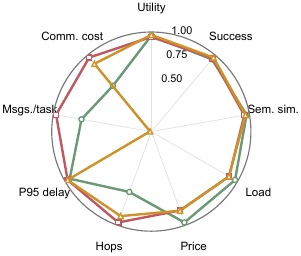}
\par\vspace{0.5mm}
{\footnotesize (a) Nominal dynamics\par}
\end{minipage}\hfill
\begin{minipage}[t]{0.49\columnwidth}
\vspace{0pt}
\centering
\includegraphics[width=\linewidth]{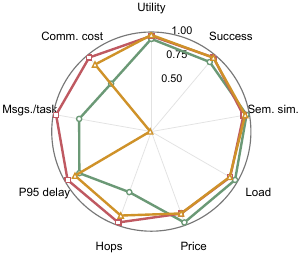}
\par\vspace{0.5mm}
{\footnotesize (b) High joint stress\par}
\end{minipage}
\caption{Normalized within-regime trade-offs under (a) nominal dynamics and
(b) high joint stress. Benefit metrics use $x/\max x$, whereas cost metrics
use $\min x/x$; larger values are better. Normalization is performed
separately within each regime.}
\label{fig:runtime-robustness}
\end{figure}

\subsection{Main Results}
Table~\ref{tab:overall-performance} shows that \SPFR maintained competitive
Utility while consistently reducing request-path overhead across all three
topologies. Based on the seed-wise paired macro-differences, \SPFR had
$1.43\%$ lower Utility than D-GREEDY (95\% CI: $0.79$--$2.07\%$ lower), but
improved Success by $1.10$ percentage points (95\% CI:
$0.59$--$1.61$ points). Among successful tasks, it reduced forward hops and
P95 delay by $34.34\%$ (95\% CI: $33.03$--$35.66\%$) and $5.79\%$
(95\% CI: $2.97$--$8.62\%$), respectively; across all tasks, it reduced
request-triggered messages by $27.24\%$ (95\% CI:
$26.02$--$28.47\%$). These directly measured endpoints, rather than the
composite SPFR potential score, independently quantify network expenditure.
Relative to GLOBAL$^\dagger$, \SPFR retained $97.63\%$ of its Utility
(95\% CI: $96.59$--$98.68\%$) while using $81.24\times$ fewer
request-triggered messages (95\% CI: $78.60$--$83.88\times$). The shared
beacon load averaged $292.03$ messages per task, yielding total signaling
loads of $295.79$, $297.20$, and $597.12$ for \SPFR, D-GREEDY, and
GLOBAL$^\dagger$, respectively. Thus, the $81.24\times$ comparison applies
specifically to request-triggered signaling.

To isolate in-path reselection, we conducted a separate static experiment on
UNINETT using 10 paired seeds, each containing 100 source-sufficient (Local)
and 100 moving-view-required (Discovery) tasks. Local tasks contained no
superior source-hidden executor. Discovery tasks instead exposed, within one
or two hops, a source-invisible executor with at least $5\%$ higher frozen
potential. All methods replayed identical tasks and FIB snapshots. The
balanced split was designed for mechanism identification rather than as an
estimate of deployment prevalence. Figure~\ref{fig:moving-view-tradeoff}
shows that \SPFR{} and SRC-FIX achieved nearly identical utility on Local
tasks ($0.831$ versus $0.832$), indicating that reselection provided no
artificial gain when the source view was already sufficient. On Discovery
tasks, however, \SPFR{} improved utility by $18.19\%$, directly isolating the
benefit of the moving local view. D-GREEDY achieved $1.18\%$ higher utility
than \SPFR{}, but incurred $13.66\%$ more hops, $7.99\%$ higher P95 delay,
$13.00\%$ more request messages, and $17.91\%$ higher communication cost.
The comparison therefore exposes the intended role of hop attenuation:
\SPFR sacrifices a small amount of executor utility to avoid pursuing
distant candidates whose marginal utility does not justify their network
cost.

\begin{figure}[t]
\centering

\begin{minipage}[t]{0.49\columnwidth}
\vspace{0pt}
\centering
\includegraphics[width=\linewidth]
{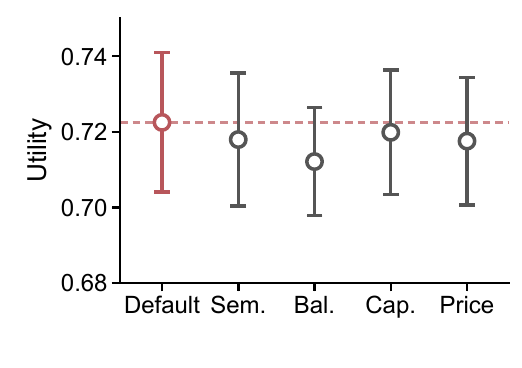}
\par\vspace{0.5mm}
{\footnotesize (a) Reference utility\par}
\end{minipage}
\hfill
\begin{minipage}[t]{0.49\columnwidth}
\vspace{0pt}
\centering
\includegraphics[width=\linewidth]
{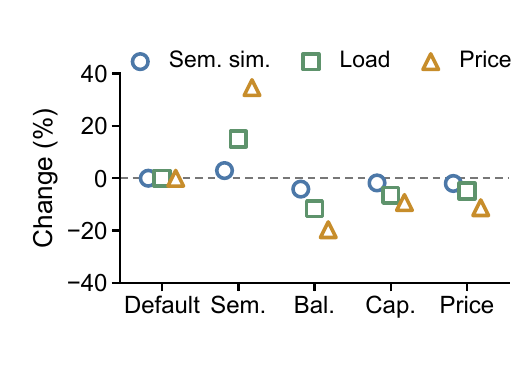}
\par\vspace{0.5mm}
{\footnotesize (b) Component shifts\par}
\end{minipage}

\vspace{-1mm}
\caption{Sensitivity to semantic, capacity, and price weights.}
\label{fig:weight-sensitivity}
\end{figure}

\begin{table}[t]
\centering
\caption{Mean seed-wise paired changes relative to full SPFR on the
independent UNINETT ablation workload.}
\label{tab:design-ablation}
\scriptsize
\setlength{\tabcolsep}{0.45pt}
\renewcommand{\arraystretch}{0.94}

\begin{tabularx}{0.97\columnwidth}{
@{}l
*{7}{>{\centering\arraybackslash}X}@{}}
\toprule
\multicolumn{1}{c}{\multirow{2}{*}{Variant}}
& $\Delta U$
& $\Delta$Succ.
& $\Delta$Msgs.
& \multirow{2}{*}{$\Delta S$}
& \multirow{2}{*}{$\Delta\rho$}
& $\Delta$Price
& $\Delta$Hops \\
& (\%)
& (pp)
& (\%)
&
&
& (c.u.)
& (\%) \\
\midrule

No reselection
& -1.10 & +0.25 & -5.2 & -0.013 & +0.005 & +0.070 & -7.1 \\

No decay
& +1.69 & -1.45 & +39.2 & +0.030 & -0.015 & -0.130 & +56.0 \\

Linear decay
& -0.21 & 0.00 & -1.2 & -0.002 & +0.001 & +0.006 & -1.7 \\

No $S$
& -6.13 & +0.35 & +3.5 & -0.110 & -0.080 & -0.636 & +4.3 \\

No $1-\rho$
& -0.22 & -0.45 & +3.7 & +0.019 & +0.036 & +0.310 & +5.4 \\

No $1-\widetilde{\pi}$
& -1.25 & -0.80 & +3.2 & +0.026 & +0.050 & +0.609 & +5.1 \\

\bottomrule
\end{tabularx}
\vspace{-1mm}
\end{table}

\subsection{Robustness, Ablations, and Sensitivity}

Fig.~\ref{fig:runtime-robustness} compares the relative trade-offs among
D-GREEDY, SPFR, and GLOBAL within each runtime regime. Because the two
panels are normalized separately, their radial magnitudes should not be used
to compare absolute degradation across regimes; the quantitative comparisons
below use the unnormalized source metrics. Under nominal dynamics, SPFR sacrificed $1.37\%$ Utility relative to
D-GREEDY, but improved success by $1.20$ percentage points. It also reduced
hops, P95 delay, request messages, and communication cost by
$2.54$--$38.17\%$. Under high joint stress, SPFR improved Utility and
success by $3.21\%$ and $4.40$ percentage points, respectively, while
retaining network-cost reductions of $14.04$--$35.25\%$. Its relative advantage therefore became
stronger when service and network states were less stable. Compared with
GLOBAL, SPFR retained at least $97.92\%$ of its utility with at least
$73.46\times$ fewer request messages. No realized deadline or budget violation was observed under the evaluated
parameter ranges; we treat these outcomes as feasibility audits rather than
evidence of performance at the SLA boundary. Overall, SPFR did not dominate
every executor attribute; it achieved a more favorable balance between
executor quality, network expenditure, and robustness.

Fig.~\ref{fig:parameter-sensitivity} shows a clear diminishing return from
expanding the control horizon. Increasing $H_{\rm ctrl}$ from 1 to 2 raised
utility from $0.589$ to $0.722$ and success from $76.6\%$ to $91.2\%$.
Increasing it further to 3 yielded only a $1.64\%$ utility gain, while FIB
state and beacon traffic increased by $2.08\times$ and $2.60\times$,
respectively. Similarly, increasing $\omega_h$ from $0.04$ to $0.08$ reduced
hops and communication cost by $14.1\%$ and $16.0\%$, with only a $0.50\%$
utility loss. We therefore selected $H_{\rm ctrl}=2$ and $\omega_h=0.08$ as
empirical operating points near the observed utility--overhead elbow, rather
than claiming that they are universally optimal.

Table~\ref{tab:design-ablation} reports mean seed-wise paired changes relative
to full \SPFR{} on the independent UNINETT ablation workload. Removing
reselection yielded $1.10\%$ lower Utility, together with $7.1\%$ fewer hops
and $5.2\%$ fewer messages. Removing hop attenuation increased Utility by
$1.69\%$, but reduced Success by $1.45$ percentage points and increased hops
and messages by $56.0\%$ and $39.2\%$, respectively. The linear ablation
replaces $\exp(-\omega_h h)$ with $[1-\omega_h h]_+$, matching the
exponential kernel in value and first-order slope at zero distance. Linear
attenuation closely matched exponential attenuation, changing Utility, hops,
and messages by only $-0.21\%$, $-1.7\%$, and $-1.2\%$, respectively.
These mean changes are consistent with attenuation primarily limiting the
network cost of pursuing distant executors. Removing semantic similarity produced the largest observed Utility loss
($6.13\%$) and reduced selected similarity by $0.110$. Removing capacity
awareness changed Utility by only $-0.22\%$, but increased selected load by
$0.036$. Removing price awareness reduced Utility by $1.25\%$ and increased
price by $0.609$~c.u. These observed component-specific shifts are consistent
with distinct roles for the three executor-utility terms.

Figure~\ref{fig:weight-sensitivity} evaluates SPFR under five
semantic/capacity/price weight profiles. These profiles are default
(70/15/15), semantic-heavy (85/7.5/7.5), balanced (50/25/25),
capacity-heavy (60/30/10), and price-heavy (60/10/30). When all outcomes
were rescored using the default weights, utility remained within
$0.712$--$0.722$. The semantic-heavy profile increased similarity from
$0.814$ to $0.838$, but also selected more loaded and expensive executors.
The balanced profile reduced load and price to $0.274$ and $1.210$~c.u.,
respectively, at lower semantic similarity. The capacity- and price-heavy
profiles produced corresponding shifts toward lower load and price. The
stable default-rescored utility and consistent component shifts show that the
observed trade-off was not specific to one weight configuration.

As a separate static measured-topology audit, we evaluated five seeds at
$H_{\rm ctrl}=3$ on GEANT, UNINETT, Deltacom, and KDL, spanning 40--754
nodes. Mean FIB size remained within 12.44--21.55 entries per agent
(P95: 20.0--44.4), whereas the aggregate beacon rate increased by
$7.15\times$ from GEANT to KDL. Thus, across the measured topologies, bounded
dissemination controlled per-agent state, while network-wide beacon traffic
remained the principal scaling cost.

\subsection{Theory and Runtime-Safety Audit}
\label{sec:theory-safety-audit}

\begin{table}[t]
\centering
\caption{Frozen-theory, dynamic-runtime, and control-plane audits.
$n$ denotes tasks for the first two scopes and events for control.}
\label{tab:theory-safety-audit}
\scriptsize
\setlength{\tabcolsep}{1.5pt}
\renewcommand{\arraystretch}{1.02}
\begin{tabularx}{\columnwidth}{
@{}
>{\centering\arraybackslash}p{0.14\columnwidth}
>{\raggedright\arraybackslash}X
>{\centering\arraybackslash}p{0.19\columnwidth}
>{\centering\arraybackslash}p{0.13\columnwidth}
@{}}
\toprule
Scope
& \multicolumn{1}{c}{Check}
& Result
& $n$ \\
\midrule

\multirow{6}{*}{Frozen}
&Bounded reference-route $\mathrm{P2Ratio}$
& 0.9623
& $6{,}000$ \\

&Full-view reference-route $\mathrm{P2Ratio}$
& 1.0000
& $3{,}000$ \\

& Maximum $\mathrm{NormP2Gap}$
& 0.651
& $5{,}996$ \\

& P2-bound violations
& 0
& $5{,}996$ \\

& Ascent/inheritance/route loops
& $0/0/0$
& $6{,}000$ \\

& Zero-attractor tasks
& $4$ ($0.067\%$)
& $6{,}000$ \\

\midrule

\multirow{2}{*}{Dynamic}
& Source-admission \NoRoute
& $351$ ($5.85\%$)
& $6{,}000$ \\

& Loops/hop-budget expirations
& $0/0$
& $4{,}000$ \\

\midrule

\multirow{3}{*}{Control}
& Converged events
& 100\%
& 120 \\

& Withdrawal/down convergence
& $3.77B$
& 60 \\

& Join/up convergence
& $0.79B$
& 60 \\

\bottomrule
\end{tabularx}
\vspace{-1mm}
\end{table}

Under frozen snapshots, at $H_{\rm ctrl}=2$ and $\omega_h=0.08$, the bounded
reference route achieved a mean $\mathrm{P2Ratio}$ of $0.9623$, while its
full-view counterpart attained $\mathrm{P2Ratio}=1$ on all $3{,}000$ tasks.
Four frozen tasks lacked a source-visible positive attractor; they were
assigned zero in the all-task ratio and excluded from the normalized-gap
audit. Among the remaining $5{,}996$ theorem-covered tasks, the maximum
$\mathrm{NormP2Gap}$ was $0.651<1$, with no P2-bound violation. No
potential-ascent, candidate-inheritance, or frozen forwarding-loop violation
was observed.

The dynamic audits are separate from, and do not extend, the frozen-state
guarantees. In the main runs, $351$ of the $6{,}000$ \SPFR{} tasks returned
\NoRoute at the source because its positive-potential eligible set was empty.
No runtime loop or hop-budget expiration occurred across the $4{,}000$
stress-audit tasks. All $120$ injected control-plane events converged;
withdrawal/down and join/up updates completed in $3.77B$ and $0.79B$ on
average, respectively, where $B$ denotes the beacon interval.

\section{Conclusion}
In this paper, we study distributed semantic task routing in IoA networks,
where executor discovery and next-hop forwarding are jointly performed under
bounded local visibility. We formulate constrained executor and path selection
and develop SPFR, which updates both decisions through task-conditioned
semantic potentials. Theoretical analysis establishes its frozen-state routing
properties, while experiments show a favorable balance between executor utility
and network overhead under dynamic conditions. SPFR performs in-path executor
discovery and reselection after forwarding is initialized by a source-visible
positive-potential executor; otherwise, it returns NoRoute. Extending route initialization to zero-attractor sources and validating
SPFR in operational IoA deployments remains future work.

\bibliographystyle{IEEEtranBST2/IEEEtran}
\bibliography{IEEEtranBST2/ArticleReferences}

\end{document}